\documentclass[11pt]{article}

\usepackage{amssymb} 
\usepackage{amsmath}     
\usepackage{amsthm}
\usepackage{todonotes}
\usepackage{mathtools}

\usepackage{a4wide}
\usepackage{graphicx}
\usepackage{hyperref}

\newtheorem{thm}{Theorem}
\newtheorem{lem}{Lemma}

\newtheorem{cla}{Claim}

\newenvironment{keyword}{\par{\noindent\bf Keywords:}}

\begin{document}

\title{Recoverable robust spanning tree problem under interval uncertainty representations} 

\author{Mikita Hradovich$^\dag$, Adam Kasperski$^\ddag$,  Pawe{\l} Zieli{\'n}ski$^\dag$\\
          {\small \textit{$^\dag$Faculty of Fundamental Problems of Technology,}}\\
	{\small \textit{Wroc{\l}aw University of Technology,  Wroc{\l}aw, Poland}}\\
	 {\small \textit{$^\ddag$Faculty of Computer Science and Management,}}\\
	{\small \textit{Wroc{\l}aw University of Technology,  Wroc{\l}aw, Poland}}\\
	{\small \texttt{\{mikita.hradovich,adam.kasperski,pawel.zielinski\}@pwr.edu.pl}}
}

 \date{}
    
\maketitle

\begin{abstract}
This paper  deals  with the recoverable robust spanning tree problem
under interval uncertainty representations.
A strongly polynomial time, combinatorial algorithm for the recoverable spanning tree problem is first constructed.  
This problem generalizes the incremental spanning tree problem, previously discussed in literature.
The algorithm built  is then applied  to solve  
 the recoverable robust spanning tree problem, under the traditional 
  interval uncertainty representation, in polynomial time.
 Moreover,  the algorithm allows to obtain
 several approximation results for the recoverable robust spanning tree problem 
 under the Bertsimas and Sim interval uncertainty representation and the
 interval uncertainty representation with a budget constraint.
 \end{abstract}


\section{Introduction}

Let $G=(V,E)$, $|V|=n$, $|E|=m$, be an undirected graph and let $\Phi$ be the set of all spanning trees of~$G$. In the \emph{minimum spanning tree problem},  a cost is specified for each edge, and we seek a spanning tree in $G$ of the minimum total cost. This problem is well known and can be solved efficiently by using several polynomial time algorithms (see, e.g.,~\cite{AMO93, PS98}). 

In this paper, we first study the
\emph{recoverable spanning tree problem}  (\textsc{Rec ST} for short). 
Namely,
for each edge $e\in E$, we are given a \emph{first stage cost}~$C_e$ and a \emph{second stage cost}~$c_e$
(recovery stage cost).
 Given a spanning tree $X\in \Phi$, let $\Phi_X^k$ be the set of all spanning trees $Y\in \Phi$ such that $|Y\setminus X|\leq k$ (the recovery set), where~$k$  is a fixed integer in $[0,n-1]$,
 called the \emph{recovery parameter}.
  Note that $\Phi_X^k$ 
 can be seen as   a neighborhood of~$X$ containing all spanning trees which can be obtained from $X$ by exchanging up to~$k$ edges. The \textsc{Rec ST} problem can be stated formally as 
 follows:
\begin{align}
\textsc{Rec ST}: \; \min_{X\in\Phi}
\left (\sum_{e\in X} C_e  +    \min_{Y\in \Phi^{k}_{X}}\sum_{e\in Y}c_e \right).
\label{recst}
\end{align}
We thus seek a first stage spanning tree $X\in \Phi$ and a second stage spanning tree $Y\in \Phi^k_X$, so that the total cost of $X$ and $Y$ for $C_e$ and $c_e$, respectively, is minimum. 
Notice that \textsc{Rec ST} generalizes the  following \emph{incremental spanning tree problem}, investigated in~\cite{SAO09}:
\begin{align}
\textsc{Inc ST}:  \min_{Y\in \Phi^{k}_{\widehat{X}}}\sum_{e\in Y}c_e,
\label{incst}
\end{align}
where $\widehat{X}\in\Phi$ is a given spanning tree.
So, we  wish to find an  improved spanning tree~$Y$ with the minimum cost, 
within a neighborhood of~$\widehat{X}$  determined by $\Phi^{k}_{\widehat{X}}$.
Several interesting practical applications of the incremental network optimization were presented in~\cite{SAO09}.
It is worth pointing out  that \textsc{Inc ST} can be seen as the \textsc{Rec ST} problem with a fixed first stage spanning tree~$\widehat{X}$, whereas in \textsc{Rec ST} both the first and the second stage trees are unknown. It has been shown in~\cite{SAO09} that \textsc{Inc ST} can be solved in strongly polynomial time by applying the Lagrangian relaxation technique. On the other hand, no strongly polynomial time combinatorial algorithm for \textsc{Rec ST} has been known to date.
Thus proposing such an algorithm for this problem is one of the main results of this paper.

The \textsc{Rec ST}  problem, beside being an interesting problem \emph{per se}, has
 an important connection with  a more general problem.
Namely, it is 
 an inner  problem in 
the \emph{recoverable robust model} with uncertain recovery costs, discussed in~\cite{B11,B12, BKK11, CG15b,LLMS09,NO13}.
Indeed, the recoverable spanning tree problem can be generalized by considering its robust version.
 Suppose that the second stage costs $c_e$, $e\in E$, are uncertain and let $\mathcal{U}$ contain all possible realizations of the second stage costs, called \emph{scenarios}. We will denote by $c_e^S$  the second stage cost of edge $e\in E$ under scenario $S\in \mathcal{U}$,  where $S=(c^{S}_e)_{e\in E}$ is a cost vector.
 In  the \emph{recoverable robust spanning tree} problem (\textsc{Rob Rec ST} for short),
 we choose an initial spanning tree $X$ in the first stage, with the cost equal to $\sum_{e\in X} C_e$. Then, after scenario $S\in\mathcal{U}$ reveals,  $X$ can be modified by exchanging at most~$k$ edges, obtaining  
 a new spanning tree $Y\in \Phi^k_{X}$. The second stage cost of $Y$ under scenario $S\in \mathcal{U}$ is 
 equal to $\sum_{e\in Y} c_e^S$. Our goal is to find a pair of trees $X$ and $Y$ such that $|X\setminus Y|\leq k$, which minimizes the sum of the first and the second stage costs $\sum_{e\in X} C_e +\sum_{e\in Y} c_e^S$ in the worst case.
The  \textsc{Rob Rec ST} problem is defined formally as follows:
\begin{equation}
\textsc{Rob Rec ST}: \; \min_{X\in \Phi} \left(\sum_{e\in X} C_e + \max_{S\in \mathcal{U}}\min_{Y\in \Phi^k_{X}} \sum_{e\in Y} c_e^S\right).
\label{rrp}
\end{equation}

If $C_e=0$ for each $e\in E$ and $k=0$, then \textsc{Rob Rec ST} is equivalent to the following \emph{min-max spanning tree} problem, examined in~\cite{ABV07,KY97,KZ11},
in which  we seek a spanning tree that
minimizes the largest cost over all scenarios:
\begin{equation}
\textsc{Min-Max ST}: \; \min_{X\in \Phi}  \max_{S\in \mathcal{U}} \sum_{e\in X} c_e^S .
\label{mrp}
\end{equation}

 If $C_e=0$ for each $e\in E$ and $k= n-1$, then \textsc{Rob Rec ST} becomes the following \emph{ 
 adversarial problem}~\cite{NO13} in which an  \emph{adversary} wants to find a scenario
 which  leads to the greatest increase in  the cost of the minimum spanning tree:
\begin{equation}
\textsc{Adv ST}: \; \max_{S\in \mathcal{U}}\min_{Y\in \Phi } \sum_{e\in Y} c_e^S.
\label{rap}
\end{equation}

We now briefly recall the known complexity results on \textsc{Rob Rec ST}. 
It turns out that 
its computational
complexity highly relies on the
 way of defining the scenario set~$\mathcal{U}$.
There are  two popular methods of representing~$\mathcal{U}$, namely the
\emph{discrete and interval uncertainty representations}.
For the \emph{discrete uncertainty representation} (see, e.g.,~\cite{KY97}), scenario set, denoted by~$\mathcal{U}^D$, contains $K$ explicitly listed scenarios, i.e. $\mathcal{U}^D=\{S_1,S_2,\dots,S_K\}$. In this case,
the \textsc{Rob Rec ST} problem is known to be NP-hard for $K=2$ and any constant $k$~\cite{KKZ14}. 
Furthermore, it becomes strongly NP-hard and not at all approximable when both $K$ and $k$ are  
a part of  the input~\cite{KKZ14}.   It is worthwhile to mention that \textsc{Min-Max ST} 
is NP hard even when $K=2$ and becomes strongly NP-hard and
not approximable within 
 $O(\log^{1-\epsilon} n)$ for any $\epsilon>0$ unless NP 
  $\subseteq$ DTIME$(n^{\mathrm{poly} \log n})$,
when $K$ is a part of input~\cite{KY97, KZ11}. 
It admits an FPTAS, when $K$ is a constant~\cite{ABV07} and is
 approximable within
$O(\log^2 n)$, when $K$ is a part of the input~\cite{KZ11}.
 The \textsc{Adv ST} problem, under scenario set $\mathcal{U}^D$, is  polynomially solvable, since it 
 boils down to
 solving $K$ traditional minimum spanning tree problems.

For  the \emph{interval uncertainty representation},  which is considered  in this paper,
one assumes that the second stage cost of each edge $e\in E$ is known to belong to the closed 
interval $[c_e, c_e+d_e]$, 
where $c_e$ is a \emph{nominal cost} of $e\in E$ and $d_e\geq 0$ is the maximum 
deviation of the cost of~$e$ from its nominal value.
In the traditional case~$\mathcal{U}$,
denoted by $\mathcal{U}^I$, is  the Cartesian product of  all these intervals~\cite{KY97},
i.e. 
\begin{equation}
\mathcal{U}^I=\{S=(c_e^S)_{e\in E}\,:\, c_e^S\in [c_e, c_e+d_e], e\in E\}.
\label{intset}
\end{equation}

In~\cite{B11} a polynomial algorithm for the \emph{recoverable robust matroid basis} problem
under scenario set~$\mathcal{U}^I$ was constructed, provided that the recovery parameter~$k$ is constant.
In consequence, \textsc{Rob Rec ST} under $\mathcal{U}^I$ is also polynomially solvable for constant $k$.
Unfortunately,  the algorithm proposed in~\cite{B11} is exponential in $k$. 
 Interestingly, the corresponding recoverable robust version of the shortest path problem 
 ($\Phi$ is replaced with the set of all $s-t$ paths in $G$) has been proven to be
  strongly NP-hard and not at all approximable even if $k=2$~\cite{B12}.
It has been recently shown in~\cite{HKZ16} that \textsc{Rob Rec ST} under $\mathcal{U}^I$ is polynomially solvable 
when~$k$ is a part of the input. In order to prove this result, 
 a technique called the \emph{iterative relaxation}
 of a linear programming formulation, whose framework was described in~\cite{LRS11}, has been applied
 This technique, however, does not imply directly a strongly polynomial algorithm for \textsc{Rob Rec ST}, since it requires the solution of a linear program. 
 
In~\cite{BS03}  a popular and commonly used modification of the scenario set~$\mathcal{U}^I$ 
has been proposed.  The new  scenario set, denoted as $\mathcal{U}^{I}_1(\Gamma)$, is 
a subset of $\mathcal{U}^I$
such that under each scenario in $\mathcal{U}^{I}_1(\Gamma)$, the costs of at most $\Gamma$ 
edges are greater than their nominal values $c_e$, where $\Gamma$ is 
assumed to be a fixed integer in $[0,m]$. Scenario set $\mathcal{U}^{I}_1(\Gamma)$ is formally defined as follows:
\begin{equation}
\mathcal{U}^{I}_1(\Gamma)=\{S=(c_e^S)_{e\in E}\,:\, c_e^S\in [c_e, c_e+\delta_ed_e], \delta_e\in \{0,1\},
e\in E, \sum_{e\in E} \delta_e\leq \Gamma\}.
\label{intset1}
\end{equation}
The parameter $\Gamma$ allows us to model the degree of uncertainty. 
 When $\Gamma=0$, then we get \textsc{Rec ST} (\textsc{Rob Rec ST} with one scenario~$S=(c_e)_{e\in E}$).
 On the other hand, when $\Gamma=m$, then we get 
 \textsc{Rob Rec ST} under
 the traditional interval uncertainty~$\mathcal{U}^I$.
 It turns out that the \textsc{Adv ST} problem under $\mathcal{U}^{I}_1(\Gamma)$  is strongly NP-hard
 (it is equivalent  to the problem of finding $\Gamma$ most vital edges)~\cite{NO13, LC93, FR99}. Consequently, the  more general \textsc{Rob Rec ST} problem is also strongly NP-hard. Interestingly, the corresponding \textsc{Min-Max ST} problem
 with~$\mathcal{U}^{I}_1(\Gamma)$ is polynomially solvable~\cite{BS03}.
 
 Yet another   interesting way of defining scenario set, which
 allows us to control the amount of uncertainty,
 is called the scenario set with a \emph{budget constraint} (see, e.g,.~\cite{NO13}). This scenario set, denoted as $\mathcal{U}^{I}_2(\Gamma)$, is defined as follows:
\begin{equation}
\mathcal{U}^{I}_2(\Gamma)=\{S=(c_e^S)_{e\in E}\,:\, c_e^S=c_e+\delta_e, \delta_e\in [0,d_e], e\in E, \sum_{e\in E} \delta_e\leq \Gamma \},
\label{intset2}
\end{equation}
where  $\Gamma\geq 0$ is a 
a fixed parameter that can be seen as a budget of an adversary, and represents the maximum total increase of the edge costs from their nominal values.
Obviously, 
 if $\Gamma$ is sufficiently large,  then $\mathcal{U}^{I}_2(\Gamma)$ reduces to the traditional interval uncertainty representation~ $\mathcal{U}^{I}$.
The computational complexity of \textsc{Rob Rec ST} for scenario set $\mathcal{U}^I_2$ is still open. We only know that its special cases, namely \textsc{Min-Max ST} and \textsc{Adv ST}, are polynomially solvable~\cite{NO13}.

In this paper we will construct a  combinatorial algorithm for   \textsc{Rec ST} with strongly polynomial running time.
We will apply  this  algorithm 
for solving   \textsc{Rob Rec ST} under scenario set~$\mathcal{U}^I$ in strongly polynomial time.
Moreover, we will show how the algorithm for  \textsc{Rec ST} can be used to
obtain several approximation results for \textsc{Rob Rec ST}, under scenario sets 
$\mathcal{U}^{I}_1(\Gamma)$ and $\mathcal{U}^{I}_2(\Gamma)$. 
This paper is organized as follows. Section~\ref{sec1} contains the main result of this paper -- a  combinatorial algorithm for \textsc{Rec ST} with strongly polynomial running time.
Section~\ref{sec2} discusses 
 \textsc{Rob Rec ST} under the interval uncertainty representations
$\mathcal{U}^{I}$,
$\mathcal{U}^{I}_1(\Gamma)$, and $\mathcal{U}^{I}_2(\Gamma)$.

\section{The recoverable spanning tree problem}
\label{sec1}

In this section we construct a combinatorial algorithm for \textsc{Rec ST} with strongly polynomial running time. Since $|X|=n-1$ for each $X\in \Phi$, \textsc{Rec ST} (see (\ref{recst})) is equivalent  to the following mathematical programming problem: 
\begin{equation}
 \begin{array}{ll}
 \min & \displaystyle  \sum_{e\in X}C_e + \sum_{e\in Y}c_e\\
\text{s.t. } & |X \cap Y| \geq L, \\
       & X,Y\in \Phi,
 \end{array}
 \label{mrecst}
\end{equation}
where $L= n-1-k$.
Problem~(\ref{mrecst}) can be expressed as the following MIP model:
\begin{eqnarray}
    Opt=\min        \sum_{e\in E}C_e x_e + \sum_{e\in E}c_e y_e&& \label{mipst_o}\\
    \text{s.t. }\;\;  \sum_{e\in E} x_e = n-1, && \label{mipst_1} \\
                      \sum_{e\in E(U)} x_e \leq  |U|-1, && \forall U\subset V, \label{mipst_2}\\
                       \sum_{e\in E} y_e = n-1, &&   \label{mipst_3}\\
                      \sum_{e\in E(U)} y_e \leq  |U|-1, && \forall U\subset V,  \label{mipst_4}\\
                      x_e-z_e \geq  0,                              &&\forall e\in E,  \label{mipst_5}\\
                       y_e-z_e  \geq  0,                              &&\forall e\in E,  \label{mipst_6}\\
                         \sum_{e\in E} z_e\geq  L, &&  \label{mipst_7}\\
                      x_e,y_e, z_e \geq 0, \text{ integer}&& \forall e\in E,  \label{mipst_8}
 \end{eqnarray}
where $E(U)$ stands for the set of edges that have both endpoints in  $U\subseteq V$. We first apply the Lagrangian relaxation (see, e.g., \cite{AMO93})
to   (\ref{mipst_o})-(\ref{mipst_8}) by relaxing the cardinality constraint~(\ref{mipst_7}) 
with a nonnegative multiplier~$\theta$. We also relax
the integrality constraints~(\ref{mipst_8}). We thus get the following linear program (with the corresponding
dual variables which will be used later):
\begin{align}
    \phi(\theta)=\min  & \sum_{e\in E}C_e x_e +\sum_{e\in E}c_e y_e- \theta\sum_{e\in E} z_e+\theta L&&& \label{llpst_o}\\
    \text{s.t. } &\sum_{e\in E} x_e = n-1,  & && [\mu], \nonumber\\
                      &-\sum_{e\in E(U)} x_e \geq  -(|U|-1),  &\forall U\subset V,& &[w_U] ,\nonumber\\
                       &\sum_{e\in E} y_e = n-1, &  &&[\nu], \nonumber\\
                      &-\sum_{e\in E(U)} y_e \geq  -(|U|-1), & \forall U\subset V, &&[v_U],  \nonumber\\
                      &x_e-z_e \geq  0,                              &\forall e\in E,  &&[\alpha_e],\nonumber\\
                       &y_e-z_e  \geq  0,                              &\forall e\in E,  &&[\beta_e],\nonumber\\
                      &x_e,y_e, z_e \geq 0, & \forall e\in E.  &&\nonumber
 \end{align}
 
For any $\theta\geq 0$, the Lagrangian function~$\phi(\theta)$ is a lower bound on~$Opt$.
It is well-known that $\phi(\theta)$ is concave and piecewise linear. 
By  the optimality test (see, e.g., \cite{AMO93}), we obtain the following theorem:
\begin{thm}
Let $(x_e, y_e, z_e)_{e\in E}$ be an optimal  solution to~(\ref{llpst_o})
for some $\theta\geq 0$, feasible to~(\ref{mipst_1})-(\ref{mipst_8}) 
 and satisfying the complementary slackness condition 
 $\theta(\sum_{e\in E} z_e-L)=0$. Then $(x_e, y_e, z_e)_{e\in E}$ is optimal to~(\ref{mipst_o})-(\ref{mipst_8}).
 \label{tloc}
\end{thm} 
Let $(X,Y)$, $X,Y\in \Phi$, be a pair of spanning trees of $G$ (a \emph{pair} for short). This pair corresponds to a feasible $0-1$ solution to~(\ref{llpst_o}), defined as follows: $x_e=1$ for $e\in X$, $y_e=1$ for $e\in Y$, and $z_e=1$ for $e\in X\cap Y$; the values of the remaining variables are set to~0. From now on, by a pair $(X,Y)$ we also mean a feasible solution to~(\ref{llpst_o}) defined as above.
Given a pair $(X,Y)$ with the corresponding solution $(x_e, y_e, z_e)_{e\in E}$, let us define the partition $(E_X, E_Y, E_Z, E_W)$ of the set of the edges $E$ in the following way: $E_X=\{e\in E\,:\, x_e=1, y_e=0\}$, $E_Y=\{e\in E\,:\, y_e=1, x_e=0\}$, $E_Z=\{e\in E \,:\, x_e=1, y_e=1\}$ 
and $E_W=\{e\in E\,:\, x_e=0,y_e=0\}$. Thus 
equalities: $X=E_X\cup E_Z$, $Y=E_Y\cup E_Z$ and $E_Z=X\cap Y$ hold. 
Our goal is to establish some sufficient optimality conditions for a given pair $(X,Y)$ in the problem~(\ref{llpst_o}).
The dual to~(\ref{llpst_o}) has the following form:
\begin{align}
    \phi^D(\theta)=\max  &-\sum_{U\subset V}(|U|-1)w_U+(n-1)\mu-\sum_{U\subset V}(|U|-1)v_U+(n-1)\nu
                                     +\theta L&&\label{dllpst_o}\\
    \text{s.t. } &-\sum_{\{U\subset V\,:\,e\in E(U)\}} w_U+\mu \leq C_e-\alpha_e,  & \forall e\in E,&\nonumber \\
                    &-\sum_{\{U\subset V\,:\,e\in E(U)\}} v_U+\nu \leq c_e-\beta_e,   &\forall e\in E,&\nonumber\\
                    &\alpha_e+\beta_e\geq \theta, &\forall e\in E,&\nonumber\\
                    &w_U,v_U\geq0, &U\subset V,&\nonumber\\
                    &\alpha_e,\beta_e\geq0, &\forall e\in E.&\nonumber
 \end{align}

\begin{lem}
\label{lemdu}
The dual problem~(\ref{dllpst_o}) can be rewritten as follows:
\begin{equation}
\label{eqdual}
\phi^D(\theta)=\max_{\{\alpha_e\geq 0, \beta_e \geq 0\,:\, \alpha_e+\beta_e\geq \theta,\; e\in E\}} \left( \min_{X\in \Phi} \sum_{e\in X} (C_e-\alpha_e)+ \min_{Y\in \Phi} \sum_{e\in Y} (c_e-\beta_e)\right)+\theta L.
\end{equation}
\end{lem}
\begin{proof}
         Fix some  $\alpha_e$ and $\beta_e$ such that $\alpha_e+\beta_e\geq \theta$ for each $e\in E$
	 in~(\ref{dllpst_o}).
	For these constant values of $\alpha_e$ and $\beta_e$, $e\in E$,
	using the dual to~(\ref{dllpst_o}), we arrive to 
         $\min_{X\in \Phi} \sum_{e\in X} (C_e-\alpha_e)+ \min_{Y\in \Phi} \sum_{e\in Y} (c_e-\beta_e)+\theta L$ and 
         the lemma    follows.
\end{proof}
Lemma~\ref{lemdu} allows us to establish the following result:
\begin{thm}[\textbf{Sufficient pair optimality conditions}]
\label{tssoc}
 A pair  $(X,Y)$  is optimal  
to~(\ref{llpst_o}) for a fixed  $\theta\geq 0$
if there exist $\alpha_e\geq 0$, $\beta_e \geq 0$ such that
 $\alpha_e+\beta_e= \theta$ for each $e\in E$ and
\begin{itemize}
 \item[(i)] $X$ is a minimum spanning tree for the costs $C_e-\alpha_e$, $Y$ is a minimum spanning tree for the costs $c_e-\beta_e,$
 \item[(ii)] $\alpha_e=0$ for each $e\in E_X$, $\beta_e=0$ for each $e\in E_Y$.
 \end{itemize}
\end{thm}
\begin{proof}
	By the primal-dual relation, the inequality $\phi^D(\theta)\leq \phi(\theta)$ holds. Using~(\ref{eqdual}), we obtain
	\begin{align*}
	\phi^D(\theta)&\geq \sum_{e\in X} (C_e-\alpha_e)+  \sum_{e\in Y} (c_e-\beta_e)+\theta L=
	\sum_{e\in E_X} C_e+\sum_{e\in E_Y}c_e+
	\sum_{e\in E_Z} (C_e+c_e-\theta)+\theta L\\
	&=\sum_{e\in E_X} C_e + \sum_{e\in E_Z} C_e+\sum_{e\in E_Y} c_e +
	\sum_{e\in E_Z} c_e - \theta|E_Z|+\theta L\\
	&=\sum_{e\in X} C_e + \sum_{e\in Y} c_e-\theta|E_Z|+\theta L=\phi(\theta).
	\end{align*}
	The Weak Duality Theorem implies the optimality of $(X,Y)$ in~(\ref{llpst_o}) for a fixed $\theta\geq 0$.
\end{proof}
\begin{lem}
     A  pair $(X,Y)$,
     which satisfies  the  sufficient pair optimality conditions for $\theta=0$,
     can be computed in polynomial 
     time.
    \label{ltheta0}
\end{lem}
\begin{proof}
	Let $X$ be a minimum spanning tree for the costs $C_e$ and $Y$
	 be  a minimum spanning tree for the costs $c_e$, $e\in E$. 
	 Since $\theta=0$, we set $\alpha_e=0$, $\beta_e=0$ for each $e\in E$. 
	 It is clear that  $(X,Y)$ satisfies  the  sufficient pair optimality conditions.
\end{proof}

Assume that  $(X,Y)$ satisfies the sufficient  pair optimality conditions for some $\theta\geq 0$.
If, for this pair,  $|E_Z|\geq L$ and $\theta(|E_Z|-L)=0$, then we are done, because by Theorem~\ref{tloc},  the pair~$(X,Y)$ is 
 optimal to~(\ref{mipst_o})-(\ref{mipst_8}).
 Suppose that $|E_Z|<L$ ($(X,Y)$ is not feasible to~(\ref{mipst_o})-(\ref{mipst_8})).
 We will now show a polynomial time procedure for finding a new pair  $(X',Y')$, which satisfies the  sufficient pair optimality conditions and $|E_{Z'}|=|E_Z|+1$. 
 This implies a polynomial time algorithm for the problem (\ref{mipst_o})-(\ref{mipst_8}), since it is enough to start with a pair satisfying  the sufficient  pair optimality conditions
 for $\theta=0$ (see Lemma~\ref{ltheta0}) and repeat the procedure at most $L$ times, i.e. until $|E_{Z'}|=L$.

Given a  spanning tree~$T$ in~$G=(V,E)$ and edge~$e=\{k,l\}\not\in T$,
let us denote by $P_T(e)$ the unique path in~$T$ connecting nodes~$k$ and~$l$. It is well known that for any $f\in P_T(e)$,  $T'=T\cup\{e\}\setminus \{f\}$ is also a spanning tree in $G$. We will say that $T'$ is the result of a \emph{move} on $T$.

Consider a pair $(X,Y)$  that satisfies the sufficient  pair optimality conditions for some fixed 
$\theta\geq 0$.  Set $C^*_e=C_e-\alpha_e$ and $c^*_e=c_e-\beta_e$ for every $e\in E$, where $\alpha_e$ and $\beta_e$, $e\in E$, are the numbers which satisfy the conditions in Theorem~\ref{tssoc}. Thus,
by Theorem~\ref{tssoc}(i) and the \emph{path optimality conditions} (see, e.g., \cite{AMO93}),
we get  the following  conditions which must be satisfied by $(X,Y)$:
\begin{subequations}
\begin{align}
   &\text{for  every }  e\notin X& &C^*_e \geq C^*_f & & \text{for  every }  f\in P_{X}(e), \label{toptca}\\
   &\text{for  every }  e\notin Y&  &c^*_e \geq c^*_f & & \text{for  every }  f\in P_{Y}(e).\label{toptcb}
\end{align}
\label{toptc}
\end{subequations}

We now build a so-called \emph{admissible graph} $G^A=(V^A, E^A)$ in two steps. 
We first  associate with each edge~$e\in E$  a node~$v_{e}$ and include it to $V^A$, $|V^A|=|E|$.
We then add arc $(v_e, v_f)$ to $E^A$ if  $e\notin X$, $f\in P_{X}(e)$ and $C^*_e=C^*_f$. This arc is called an \emph{$X$-arc}.
We also add arc  $(v_f,v_e)$ to $E^A$ if  $e\notin Y$, $f\in P_{Y}(e)$ and $c^*_e=c^*_f$. This arc is called an \emph{$Y$-arc}.
We say that $v_e\in V^A$ is  \emph{admissible} if
  $e\in E_Y$, or $v_e$ is reachable from a node $v_g\in V^A$, such that  $g\in E_Y$,
  by a directed path in $G^A$. 
  In the second step
  we remove from $G^A$ all the nodes which are not admissible, together with their incident arcs. 
  An  example of an admissible graph is shown in Figure~\ref{fga}.  Each node of this admissible graph is reachable from some node $v_g$, $g\in E_Y$. Note that the arcs $(v_{e_7}, v_{e_6})$  and $(v_{e_7}, v_{e_{10}})$ are not present in $G^A$, because $v_{e_7}$ is not reachable from any node $v_g$, $g\in E_Y$. These arcs have been removed from $G^A$ in the second step.
  
  Observe that each $X$-arc $(v_e, v_f)\in E^A$ represents a move on $X$, namely $X'=X\cup \{e\} \setminus \{f\}$ is a spanning tree in $G$. Similarly, each $Y$-arc $(v_e,v_f)\in E^A$ represents a move on $Y$, namely $Y'=Y\cup\{f\} \setminus \{e\}$ is a spanning tree in $G$.  Notice that the cost, with respect to $C^*_e$, of $X'$  is the same as $X$ and the cost, with respect to $c^*_e$, of $Y'$  is the same as $Y$. So, the moves indicated by $X$-arcs and $Y$-arcs preserve the optimality of $X$ and $Y$, respectively.
Observe that $e\notin X$ or $e\in Y$, which implies $e\notin E_X$. Also $f\in X$ or $f\notin Y$, which implies $f\notin E_Y$. 
  Hence, no arc in $E^A$ can start in a node corresponding to an edge in $E_X$ and no arc in $E^A$ can end in a node corresponding to an edge in $E_Y$. Observe also that $(v_e,v_f)\in E^A$ can be both $X$-arc and $Y$-arc only if $e\in E_Y$ and $f\in E_X$. Such a case is shown in Figure~\ref{fga} (see the arc $(v_{e_1},v_{e_2})$). 
  Since each arc  $(v_e,v_f)\in E^A$ represents a move on $X$ or $Y$, $e$ and $f$ cannot both belong to $E_W$ or $E_Z$.
 \begin{figure}[h]
\centering
\includegraphics[width=13.3cm]{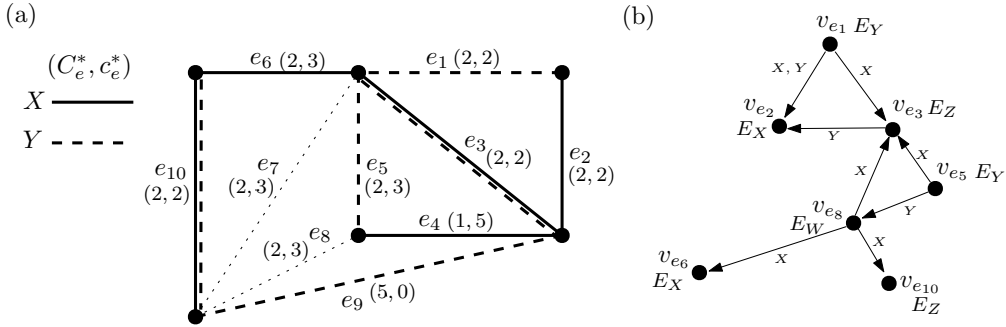}
\caption{(a) A pair $(X,Y)$ such that $X=\{e_2,e_3,e_4,e_6,e_{10}\}$ and $Y=\{e_1,e_3,e_5,e_9,e_{10}\}$. 
(b)~The admissible graph $G^A$ for $(X,Y)$.} \label{fga}
\end{figure}
    
We will consider two cases: $E_X\cap \{e\in E\,:\, v_e\in V^A\}  \neq \emptyset$ and
$E_X\cap \{e\in E\,:\, v_e\in V^A\}  = \emptyset$. The first case means that there is a directed path from $v_e$, $e\in E_Y$, to a node $v_f$, $f\in E_X$, in the admissible graph $G^A$ and in the second case no such a path exists.  We will show that in the first case it is possible to find a new pair $(X', Y')$ which satisfies the  sufficient pair optimality conditions and $|E_{Z'}|=|E_{Z}|+1$. The idea will be to perform a sequence of moves on $X$ and $Y$, indicated by the arcs on some suitably chosen path from $v_e$, $e\in E_Y$, to $v_f$, $f\in E_X$ in the admissible graph $G^A$. Let us formally handle this case.
 \begin{lem}
  \label{tXVA}
   If $E_X\cap \{e\in E\,:\, v_e\in V^A\}  \neq \emptyset$, then there exists a pair
   $(X', Y')$ with $|E_{Z'}|=|E_Z|+1$, which satisfies the sufficient  pair optimality conditions for $\theta$.
  \end{lem}
  \begin{proof}  
  We begin by introducing the notion of a \emph{cycle graph} $G(T)=(V^T,A^T)$,
  corresponding to a given spanning tree~$T$ of graph~$G=(V,A)$.
  We build~$G(T)$ as follows: we associate 
   with each edge~$e\in E$  a node~$v_{e}$ and include it to $V^T$, $|E|=|V^T|$;
   then we add arc $(v_e,v_f)$ to $A^T$ if $e\not \in T$ and $f\in P_{T}(e)$.
  An example is shown in Figure~\ref{fcyclegraph}.
\begin{figure}[h]
\centering
\includegraphics{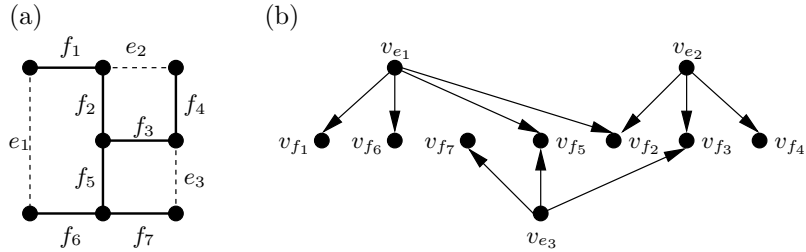}
\caption{(a) A graph $G$ with a spanning tree $T$ (the solid lines). (b) The cycle graph $G(T)$.} \label{fcyclegraph}
\end{figure}

\begin{cla}
\label{ccycleg}
Given a spanning tree $T$ of~$G$, 
let $\mathcal{F}=\{(v_{e_1},v_{f_1}), (v_{e_2},v_{f_2}),\dots, (v_{e_\ell}, v_{f_\ell})\}$ be a subset of arcs of $G(T)$, where all $v_{e_i}$ and  $v_{f_i}$  (resp. $e_i$ and $f_i$), $i\in  [\ell]$,  are distinct. If 
$T'=T\cup \{e_1, \dots,e_\ell\} \setminus \{f_1,\dots,f_\ell\}$ is not a spanning tree, then $G(T)$ contains a subgraph depicted in Figure~\ref{fcycle}, where $\{j_1,\dots,j_\kappa\}\subseteq [\ell]$.	
\end{cla}  

 \begin{figure}[h]
\centering
\includegraphics{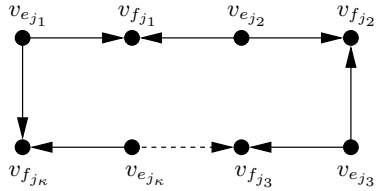}
\caption{A subgraph of $G(T)$ from Claim~\ref{ccycleg}.} \label{fcycle}
\end{figure}
Let us illustrate Claim~\ref{ccycleg} by using the sample graph in Figure~\ref{fcyclegraph}. 
Suppose that 
$\mathcal{F}=\{(v_{e_1},v_{f_5}), (v_{e_2},v_{f_2}), (v_{e_3},v_{f_3})\}$. 
Then $T'=T\cup \{e_1,e_2,e_3,\} \setminus \{f_5, f_2, f_3\}$ is not a spanning tree and $G(T)$ contains the  subgraph composed of the following arcs (see Figure~\ref{fcyclegraph}): 
\[
(v_{e_1},v_{f_2}),(v_{e_2},v_{f_2}), (v_{e_2},v_{f_3}),(v_{e_3},v_{f_3}),
(v_{e_3},v_{f_5}),(v_{e_1},v_{f_5}).
\]

\begin{proof}[Proof of Claim~\ref{ccycleg}]
We form $T'$ by performing a sequence of moves consisting in adding edges~$e_i$ and removing
edges $f_i\in P_{T}(e_i)$, $i\in [\ell]$. 
Suppose that, at some step, a cycle appears, which is formed by some edges from $\{e_1,\dots,e_{\ell}\}$ and 
the remaining edges of~$T$
(not removed  from $T$). Such a cycle must appear, since otherwise $T'$ would be a spanning tree.
Let us relabel the edges so that $\{e_1,\dots,e_s\}$ are on this cycle, i.e. the first $s$ moves consisting in adding $e_i$ and removing $f_i$ create the cycle, $i\in[s]$. An example of such a situation for $s=4$ is shown in Figure~\ref{fclaimgraph}. The cycle is formed by the edges $e_1,\dots,e_4$ and the paths  $P_{v_2v_3}$, $P_{v_4v_5}$ and $P_{v_1v_6}$ in $T$. Consider the edge $e_1=\{v_1,v_2\}$.  Because $T$ is a spanning tree, $P_T(e_1)\subseteq P_{v_2v_3}\cup P_T(e_2) \cup P_T(e_3)\cup P_{v_4v_5}\cup P_T(e_4)\cup P_{v_1v_6}$. Observe that $f_1\in P_T(e_1)$ cannot belong to any of $P_{v_2v_3}$, $P_{v_4v_5}$ and $P_{v_1v_6}$. If it would be contained in one of these paths, then no cycle would be created.
Hence, $f_1$ must belong to $P_T(e_2)\cup P_T(e_3)\cup P_T(e_4)$.
The above argument is general and, by using it, we can show that for each $i\in [s]$, $f_i\in P_T(e_j)$ for some $j\in [s]\setminus \{i\}$.  
\begin{figure}[h]
\centering
\includegraphics{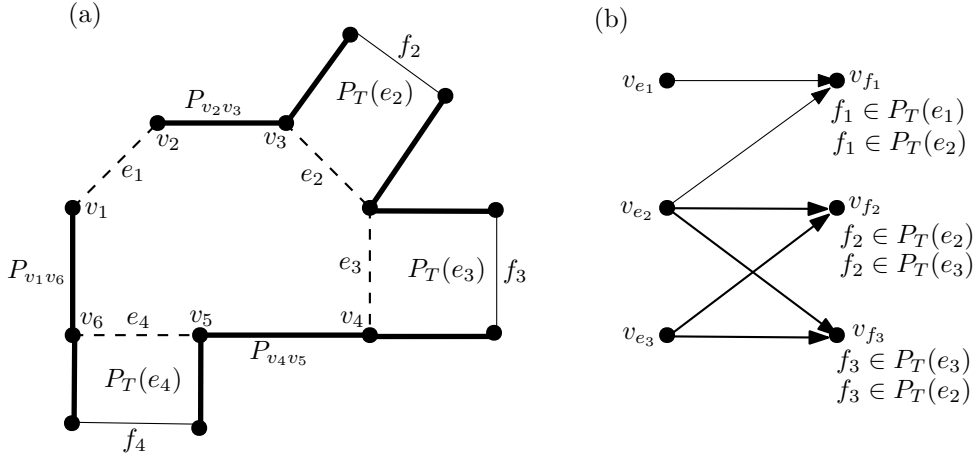}
\caption{Illustration for the proof of Claim~\ref{ccycleg}. (a) The bold lines represent paths in $T$ (not necessarily disjoint); $f_1\in P_T(e_2)\cup P_T(e_3)\cup P_T(e_4)$. (b) The subgraph $G'(T)$ with the corresponding cycle.} \label{fclaimgraph}
\end{figure}

We are now ready to build a subgraph depicted in Figure~\ref{fcycle}.  Consider a subgraph $G'(T)$ of the cycle graph $G(T)$ built as follows. The nodes of $G'(T)$ are $v_{e_1},\dots,v_{e_s},v_{f_1},\dots,v_{f_s}$. Observe that $G'(T)$ has exactly $2s$ nodes, since all the edges $e_1,\dots,e_s, f_1,\dots,f_s$ are distinct by the assumption of the claim. For each $i\in [s]$ we add to $G'(T)$ two arcs, namely $(v_{e_i},v_{f_i})$, $f_i\in P_T(e_i)$ and $(v_{e_j}, v_{f_i})$, $f_i\in P_T(e_j)$ for $j\in [s]\setminus \{i\}$ (see Figure~\ref{fclaimgraph}). The resulting graph $G'(T)$ is bipartite and has exactly $2s$ arcs. In consequence $G'(T)$ (and thus $G(T)$) must contain a cycle which is of the form depicted in Figure~\ref{fcycle}. 
\end{proof}
 
After this preliminary step, we can now return to
the main proof.    
  If $E_X\cap \{e\in E\,:\, v_e\in V^A\}  \neq \emptyset$,
   then, by the construction of the admissible graph, there exists a directed path in $G^A$ from a node~$v_e$,  $e\in E_Y$, 
    to a node~$v_f$, $f\in E_X$.
  Let $P$ be a shortest such a path from $v_e$ to $v_f$, i.e. a path consisting of the fewest number of arcs,
  called an \emph{augmenting path}.
   We need to consider the following cases:
  \begin{enumerate}
  \item 
  \label{auc1}
  The augmenting path~$P$ is of the form:
$$\begin{array}{lcl}
	E_Y  & & E_X\\
	v_e& \rightarrow &  v_f
	\end{array}$$
	If $(v_e,v_f)$ is $X$-arc,  then  $X'=X\cup \{e\} \setminus \{f\}$ is an updated spanning tree of~$G$ such that $|X'\cap Y|=|E_Z|+1$. Furthermore $X'$ is a minimum spanning tree for the costs $C^*_e$ and the new pair $(X',Y)$ satisfies the sufficient pair optimality conditions ($E_{X'}\subseteq E_{X}$, so condition (ii) in Theorem~\ref{tssoc} is not violated). If $(v_e,v_f)$ is $Y$-arc, then $Y'=Y\cup \{f\} \setminus \{e\}$ is an updated spanning tree of~$G$ such that $|X\cap Y'|=|E_Z|+1$. Also $Y'$ is a minimum spanning tree for the costs $c^*_e$ and the new pair $(X,Y')$ satisfies the sufficient pair optimality conditions.
An example can be seen in Figure~\ref{fga}. There is a path $v_{e_1}\rightarrow v_{e_2}$ in the admissible graph. The arc $(v_{e_1},v_{e_2})$  is both $X$-arc and $Y$-arc. We can thus choose one of the two possible moves $X'=X\cup\{e_1\}\setminus \{e_2\}$ or $Y'=Y\cup\{e_2\} \setminus \{e_1\}$, which results in $(X',Y)$ or $(Y', X)$.
	
\item
\label{auc2}
 The augmenting path~$P$ is of the form:
	$$
	\arraycolsep=3.5pt
	\begin{array}{lllllllllllllllllllll}
	& E_Y  & & E_Z & & E_W & & E_Z & & E_W & & E_Z & & & &  E_W & & E_Z & & E_X\\
	(a) & v_{e_1} & \stackrel{X}{\rightarrow} &  v_{f_1} &  \stackrel{Y}{\rightarrow}  &  v_{e_2} & \stackrel{X}{\rightarrow}  & v_{f_2} &  \stackrel{Y}{\rightarrow}  & v_{e_3} & \stackrel{X}{\rightarrow}  &  v_{f_3} & \stackrel{Y}{\rightarrow}  & \cdots &\stackrel{Y}{\rightarrow}   & v_{e_\ell} & \stackrel{X}{\rightarrow}  & v_{f_\ell}  & \stackrel{Y}{\rightarrow}  & v_{e_{\ell+1}} \\
	   &  & &   & & & & & & & &&&& & & & E_X  \\
	   (b) & & &   & & & & & & & &&&& & &  \stackrel{X}{\rightarrow}  &   v_{f_\ell}
	\end{array}$$
	Let $X'=X\cup\{e_1,\dots,e_\ell\}\setminus\{f_1,\dots,f_\ell\}$. Let $Y'=Y\cup\{e_2,\dots,e_{\ell+1}\}\setminus\{f_1,\dots f_\ell\}$ for case $(a)$, and $Y'=Y\cup\{e_2,\dots,e_\ell\}\setminus\{f_1,\dots f_{\ell-1}\}$ for case $(b)$. 
	We now have to show that the resulting pair $(X', Y')$ is a pair of spanning trees. Suppose that $X'$ is not a spanning tree. Observe that the $X$-arcs
$(v_{e_1},v_{f_1}),\ldots,(v_{e_\ell},v_{f_\ell})$ belong to the cycle graph $G(X)$. Thus, by Claim~\ref{ccycleg},
 the cycle graph $G(X)$ must contain a subgraph depicted in Figure~\ref{fcycle}, where $\{j_1,\dots,j_\kappa\}\subseteq [\ell]$. An easy verification shows that all edges $e_i, f_i$, $i\in \{j_1,\dots,j_\kappa\}$ must have the same costs with respect to $C^*_e$.
 Indeed, if some costs are different, then there exists an edge exchange which decreases the cost of $X$.
 This contradicts our assumption that~$X$ is a minimum spanning tree with respect to $C^*_e$.
  Finally, there must be an arc $(v_{e_{i'}}, v_{f_{i''}})$ in the subgraph such that $i'<i''$. 
  Since $C^*_{e_{i'}}=C^*_{f_{i''}}$, the arc
  $(v_{e_{i'}}, v_{f_{i''}})$ is present in the admissible graph~$G^A$.
  This leads to a contradiction with our assumption that $P$ is an augmenting path. Now suppose that $Y'$ is not a spanning tree. 
  We consider only the case $(a)$ since the proof of case $(b)$ is just the same. For a convenience, let us number the nodes $v_{e_i}$ on $P$ from $i=0$ to $\ell$, so that $Y'=\{e_1,\dots,e_{\ell}\}\setminus \{f_1,\dots,f_{\ell}\}$. The arcs $(v_{e_1}, v_{f_1}), \dots, (v_{e_{\ell}},v_{f_{\ell}})$, which correspond to the $Y$-arcs $(v_{f_1},v_{e_1}),\dots, (v_{f_{\ell}},v_{e_{\ell}})$ of $P$, belong to the cycle graph $G(Y)$. Hence, by Claim~\ref{ccycleg}, $G(Y)$  must contain a subgraph depicted in Figure~\ref{fcycle}, where $\{i_1,\dots,i_{\kappa}\}\subseteq [\ell]$. The rest of the proof is similar to the proof for $X$. Namely, the edges $e_i$ and $f_i$ for $i\in \{i_1,\dots,i_{\kappa}\}$ must have the same costs with respect to $c^*_e$. Also, there must exist an arc $(v_{e_{i'}}, v_{f_{{i''}}})$ in the subgraph such that $i'>i''$. In consequence, the arc $(v_{f_{i''}}, v_{e_{i'}})$ belongs to the admissible graph, which contradicts the assumption that $P$ is an augmenting path.
  
	 An example of the case $(a)$ is shown in Figure~\ref{figex2}. Thus $X'=X\cup\{e_1,e_2,e_3,e_4\}\setminus\{f_1,f_2,f_3,f_4\}$ and $Y'=Y\cup\{e_2,e_3,e_4,e_5\}\setminus\{f_1,f_2,f_3,f_4\}$. An example of the case (b) is shown in Figure~\ref{figex3}. In this example $X'$ is the same as in the previous case and $Y'=Y\cup\{e_2,e_3,e_4\}\setminus\{f_1,f_2,f_3\}$.

It is easy to verify that $|E_{Z'}|=|X'\cap Y'|=|E_{Z}|+1$ holds (see also the examples in Figures~\ref{figex2} and~\ref{figex3}). The spanning trees  $X'$  and $Y'$ are optimal for the costs $C^*_e$ and $c^*_e$, respectively. Furthermore, $E_{X'}\subseteq E_X$ and $E_{Y'}\subseteq E_Y$, so $(X',Y')$ 
satisfies the sufficient pair  optimality conditions (the condition (ii) in Theorem~\ref{tssoc} is not violated).
\begin{figure}[h]
\centering
\includegraphics{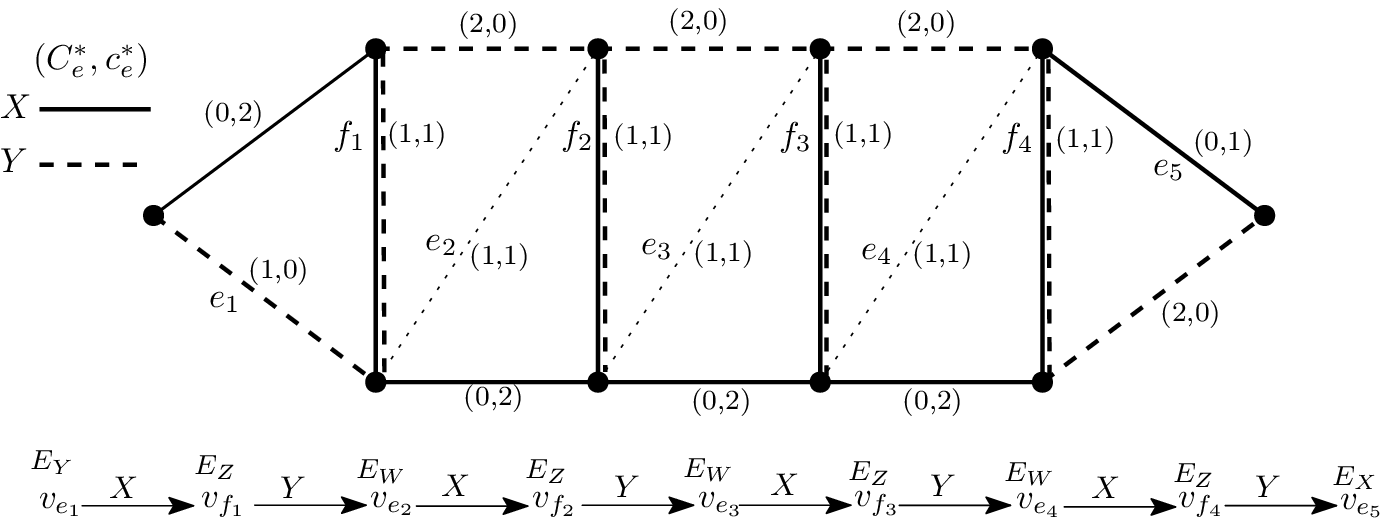}
\caption{A  pair $(X,Y)$ and the corresponding admissible graph for 
             the case~\ref{auc2}a.} \label{figex2}
\end{figure}
  \begin{figure}[h]
\centering
\includegraphics{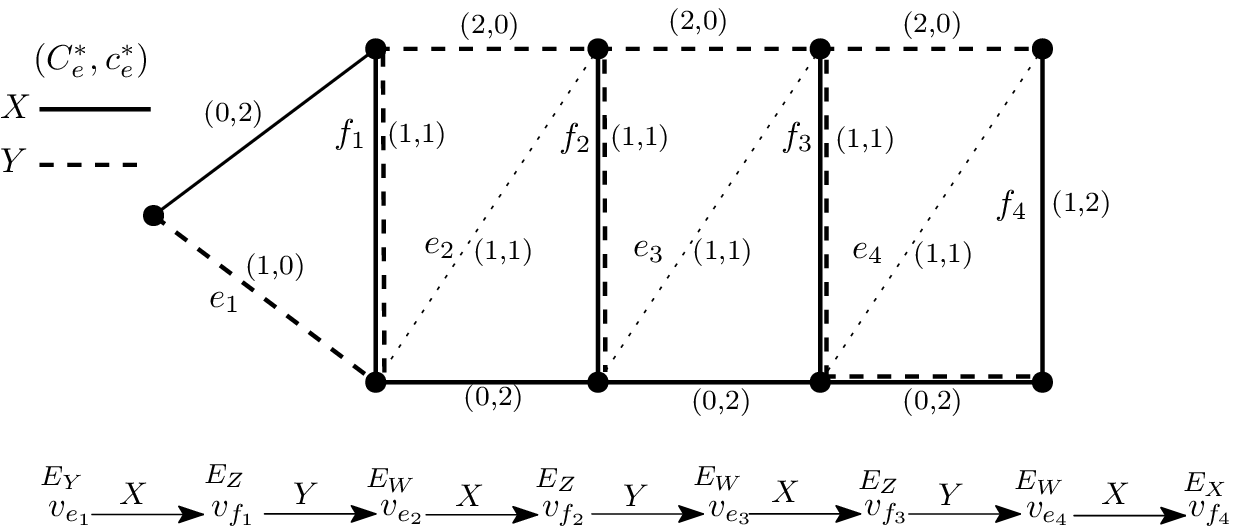}
\caption{A pair $(X,Y)$ and the corresponding admissible graph for the case~\ref{auc2}b.} \label{figex3}
\end{figure}

\item
\label{auc3}
The augmenting path $P$ is of the form
	$$
	\arraycolsep=3.5pt
	\begin{array}{llllllllllllllllllllll}
	& E_Y  & & E_W & & E_Z & & E_W & & E_Z & & E_W & & & &  E_Z & & E_W & & E_X\\
	(a) & v_{e_1} & \stackrel{Y}{\rightarrow} &  v_{f_1} &  \stackrel{X}{\rightarrow} &  v_{e_2} & \stackrel{Y}{\rightarrow} & v_{f_2} &  \stackrel{X}{\rightarrow} &  v_{e_3} & \stackrel{Y}{\rightarrow} &  v_{f_3} & \stackrel{X}{\rightarrow} & \dots & \stackrel{X}{\rightarrow}  & v_{e_\ell} & \stackrel{Y}{\rightarrow} & v_{f_\ell}  &\stackrel{X}{\rightarrow}& v_{e_{\ell+1}} \\
	    & & &   & & & & & & & &&&& & & & E_X  \\
	    (b) & & &   & & & & & & & &&&& & &  \stackrel{Y}{\rightarrow} &   v_{f_\ell}
	\end{array}$$
	Let $X'=X\cup\{f_1,\dots,f_\ell\}\setminus\{e_2,\dots,e_{\ell+1}\}$ for the case $(a)$ and
	$X'=X\cup\{f_1,\dots,f_{\ell-1}\}\setminus\{e_2,\dots e_\ell\}$ for the case $(b)$. Let $Y'=Y\cup\{f_1,\dots,f_\ell\}\setminus\{e_1,\dots e_\ell\}$. The proof that $X'$ and $Y'$ are spanning trees follows by the same arguments as for the symmetric case described in point~\ref{auc2}.
	An example of the case $(a)$ is shown in Figure~\ref{figex4}. Thus $X'=X\cup\{f_1,f_2,f_3,f_4\}\setminus\{e_2,e_3,e_4,e_5\}$ and $Y'=Y\cup\{f_1,f_2,f_3,f_4\}\setminus\{e_1,e_2,e_3,e_4\}$. An example for the case $(b)$ is shown in Figure~\ref{figex5}. The spanning tree $Y'$ is the same as in the previous case and $X'=X\cup\{f_1,f_2,f_3\}\setminus\{e_2,e_3,e_4\}$. 
		The equality $|E_{Z'|}=|X'\cap Y'|=|E_Z|+1$ holds. Also, the trees $X'$  and $Y'$ are optimal for the costs $C^*_e$ and $c^*_e$, respectively, $E_{X'}\subseteq E_X$, $E_{Y'}\subseteq E_Y$, so $(X',Y')$  satisfies the sufficient pair  optimality conditions.
\begin{figure}[h]
\centering
\includegraphics{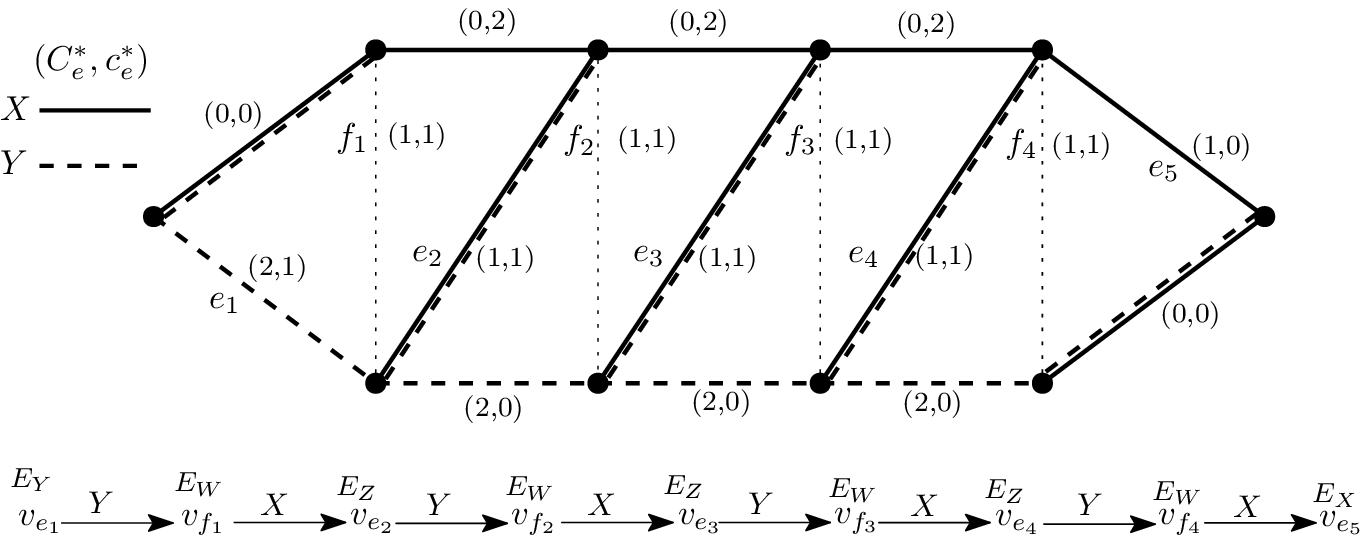}
\caption{A  pair $(X,Y)$ and the corresponding admissible graph for the case~\ref{auc3}a.} \label{figex4}
\end{figure}
\begin{figure}[h]
\centering
\includegraphics{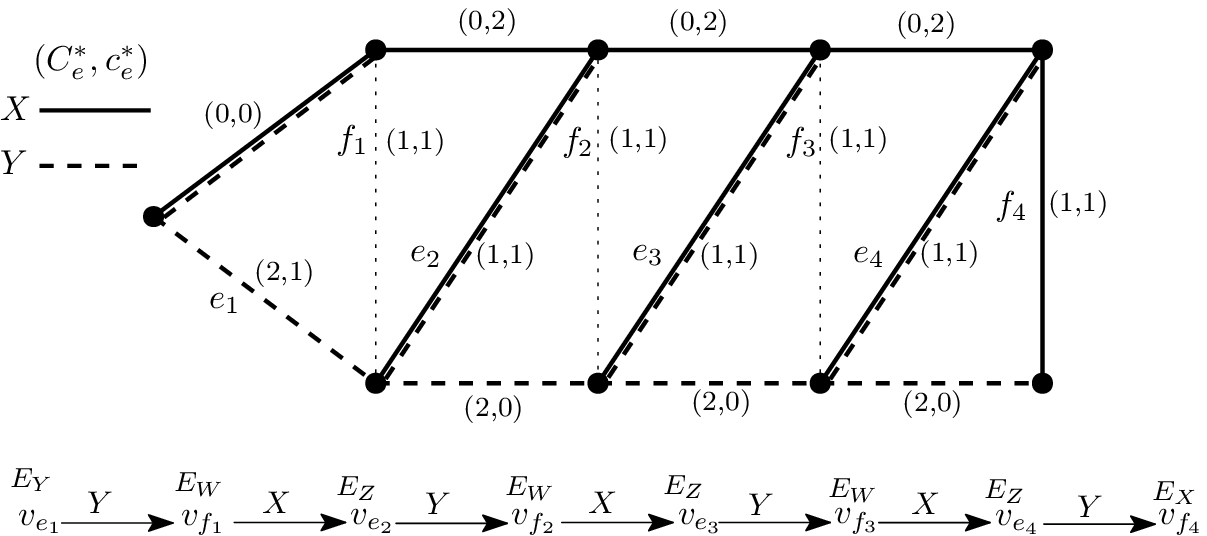}
\caption{A  pair $(X,Y)$ and the corresponding admissible graph for the case~\ref{auc3}b.} \label{figex5}
\end{figure}
 \end{enumerate}
 \end{proof} 

We now turn to the case $E_X\cap \{e\in E\,:\, v_e\in V^A\}  = \emptyset$.
Fix $\delta>0$
(the precise value of $\delta$ will be specified later) and set:
\begin{subequations}
\begin{align}
   &C_e(\delta)=C^*_e-\delta,  &&c_e(\delta)=c^*_e          && v_e \in V^A, \label{newca} \\
   &C_e(\delta)=C^*_e,            &&c_e(\delta)=c^*_e-\delta && v_e\notin V^A. \label{newcb} 
\end{align}
\label{newc}
\end{subequations}
\begin{lem}
\label{lemdelta}
  There exists a sufficiently small $\delta>0$ such that
   the costs $C_e(\delta)$ and $c_e(\delta)$ 
	satisfy the path optimality conditions for $X$ and $Y$, respectively, i.e:
\begin{subequations}
\begin{align}
   &\text{for  every }  e\notin X& &C_e(\delta) \geq C_f(\delta) & & \text{for  every }  f\in P_{X}(e), \label{toptcsa}\\
   &\text{for  every }  e\notin Y&  &c_e(\delta) \geq c_f(\delta) & & \text{for  every }  f\in P_{Y}(e).\label{toptcsb}
\end{align}
\label{toptcs}
\end{subequations}
\end{lem}
\begin{proof}
	If $C^*_e>C^*_f$  (resp.  $c^*_e >c^*_f$), $e\notin X, f\in P_{X}(e)$ (resp. $e\notin Y, f\in P_{Y}(e)$), 
	then there is $\delta>0$, such that after setting the new costs~(\ref{newc}) the inequality
	$C_e(\delta)\geq C_f(\delta)$ (resp. $c_e(\delta) \geq c_f(\delta) $) holds. Hence, one can choose  a sufficiently small~$\delta>0$ such that after setting the new costs~(\ref{newc}),
	all the strong inequalities 
	are not violated.
		Therefore,  for such a chosen~$\delta$
	it remains to show that all originally tight inequalities in~(\ref{toptc}) are preserved for the new costs.  
	Consider a tight inequality of the form: 
      	 \begin{equation}
		\label{ep1}	
		C^*_e=C^*_f, \,e\notin X, \,f\in P_{X}(e).
	\end{equation}
On the contrary, suppose that  $C_e(\delta)<C_f(\delta)$.  This is only possible when $C_e(\delta)=C^*_e-\delta$ and $C_f(\delta)=C^*_f$. Hence and from the construction of the new costs, we have
$v_f\notin V^A$ (see (\ref{newcb}))  and  $v_e\in V^A$ (see (\ref{newca})).  
By~(\ref{ep1}), we obtain $(v_e,v_f)\in E^A$. Thus $v_f\in V^A$, a contradiction.
Consider a tight inequality of the form:
\begin{equation}
\label{ep2}
c^*_e=c^*_f, \, e\notin Y,\, f\in P_{Y}(e).
\end{equation}
 On the contrary, suppose that $c_e(\delta)<c_f(\delta)$. This is only possible when 
 $c_e(\delta)=c^*_e-\delta$ and $c_f(\delta)=c^*_f$. Thus 
 we deduce that  $v_e\notin V^A$ and $v_f\in V^A$ (see (\ref{newc})).
  From~(\ref{ep2}), it follows that $(v_f,v_e)\in E^A$ and so $v_e\in V^A$,  a contradiction.
 \end{proof}
 We are now ready to give the precise value of $\delta$. We do this by 
  increasing the value of~$\delta$ until  some inequalities, originally not tight in~(\ref{toptc}),   become tight.
  Namely,
  let $\delta^*>0$ be the  smallest value of $\delta$ for which an inequality originally not tight becomes tight.
  Obviously, it occurs when $C^*_e-\delta^*=C^*_f$ for $e\notin X$, $f\in P_{X}(e)$ or $c^*_f-\delta^*=c^*_e$ for $f\notin Y$, $e\in P_{Y}(f)$.  By~(\ref{newc}),  $v_e\in V^A$ and $v_f\notin V^A$.
  Accordingly,
   if $\delta=\delta^*$, then at least one arc
   is added to $G^A$.
    Observe also that no arc can be removed from $G^A$ - the admissibility of the nodes remains unchanged.   
     It follows from the fact that each tight inequality
    for $v_e\in V^A$ and $v_f\in V^A$ is still tight. 
  This leads to the following lemma.
\begin{lem}
\label{linttheta}
	If $E_X\cap \{e\in E\,:\, v_e\in V^A\}  = \emptyset$, then $(X, Y)$ satisfies the sufficient  pair optimality conditions for each $\theta' \in [\theta, \theta+\delta^*]$.
\end{lem}
\begin{proof}
Set $\theta'=\theta+\delta$, $\delta\in [0, \delta^*]$. Lemma~\ref{lemdelta} 
implies that $X$ is optimal for $C_e(\delta)$ and $Y$ is optimal for $c_e(\delta)$. 
From~(\ref{newc}) and the definition of the costs $C^*_e$ and $c^*_e$, it follows 
that $C_e(\delta)=C_e-\alpha'_e$ and $c_e(\delta)=c_e-\beta'_e$, where $\alpha_e'=\alpha_e+\delta$
 and $\beta'_e=\beta_e$ for each $v_e\in V^A$, $\alpha'_e=\alpha_e$ and $\beta'_e=\beta_e+\delta$ for
  each $v_e\notin V^A$.  Notice that $\alpha'_e+\beta'_e=\alpha_e+\beta_e+\delta=\theta+\delta=\theta'$ 
  for each $e\in E$. By~(\ref{newc}), $c_e(\delta)=c_e$ for each $e\in E_Y$ (recall that $e\in E_Y$ implies $v_e\in V^A$), and thus $\beta_e=0$ for each $e\in E_Y$. Since $E_X\cap \{e\in E\,:\, v_e\in V^A\}  = \emptyset$, $C_e(\delta)=C^*_e=C_e$  holds 
  for each $e\in E_X$,  and so $\alpha_e=0$ for each $e\in E_X$.
  We thus have shown that there exist $\alpha'_e,\beta'_e \geq 0$ such that
 $\alpha'_e+\beta'_e= \theta'$ for each $e\in E$ satisfying  the conditions~(i) and~(ii) in Theorem~\ref{tssoc},
 which completes the proof.
\end{proof}

We now describe a polynomial procedure
that,   for a given pair $(X,Y)$ satisfying  the sufficient  pair optimality conditions for some $\theta\geq 0$,
finds a new pair of spanning trees $(X',Y')$, 
 which also satisfies the  sufficient pair optimality conditions with $|E_Z'|=|E_Z|+1$.
We start by building the admissible graph $G^A=(V^A,E^A)$ for $(X,Y)$. If this graph contains an augmenting 
 path, 
 then by Lemma~\ref{tXVA}, we are done.
  Otherwise,
 we determine $\delta^*$ and modify the costs by using~(\ref{newc}). Lemma~\ref{linttheta} shows that
  $(X,Y)$ satisfies the sufficient  pair optimality conditions for $\theta+\delta^*$. 
  For $\delta^*$ some new arcs are  added to the admissible graph $G^A$
  (all the previous arcs must be still present in $G^A$). Thus $G^A$ is  updated and 
  we set $C^*_e:=C_e(\delta^*)$, $c^*_e:=c_e(\delta^*)$ for each $e\in E$, and $\theta:=\theta+\delta^*$.  
  We repeat this until there is an augmenting  path in $G^A=(V^A,E^A)$.
  Note that such a path must appear after at most $m=|E|$ iterations, which follows from the fact that at some step a node $v_e$ such that $e\in E_X$ must appear in $G^A$.
  \begin{figure}[h]
\centering
\includegraphics[width=10cm]{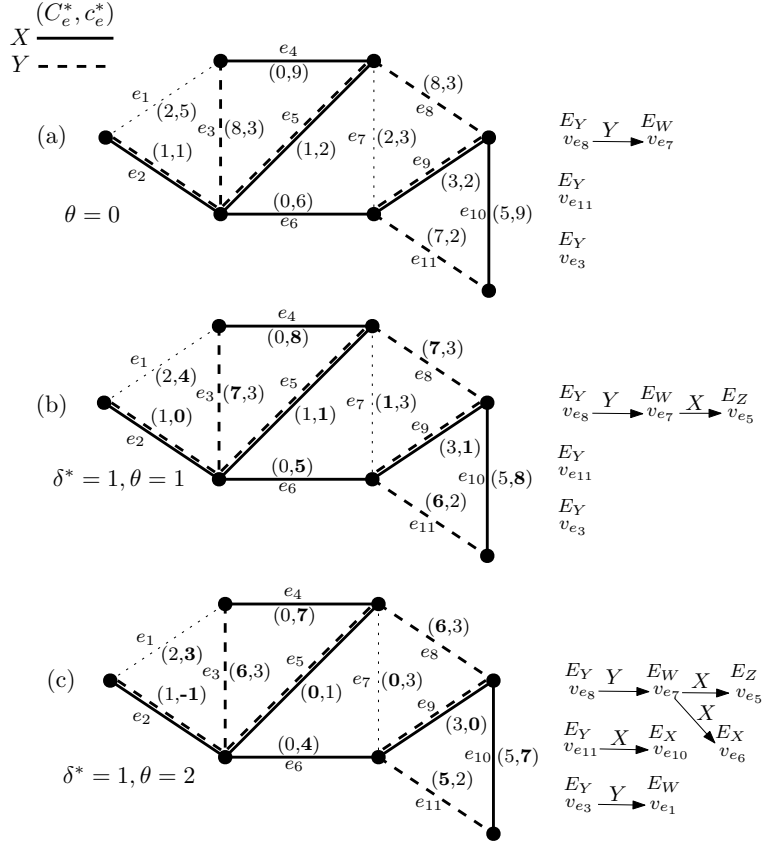}
\caption{Sample computations, $X=\{e_2, e_4,e_5,e_6,e_9,e_{10}\}$ and $Y=\{e_2, e_3, e_5, e_8, e_9, e_{11}\}$.} \label{figexiter}
\end{figure}

Sample computations are shown in Figure~\ref{figexiter}. We start with the pair $(X,Y)$, where $X=\{e_2, e_4,e_5,e_6,e_9,e_{10}\}$ and $Y=\{e_2, e_3, e_5, e_8, e_9, e_{11}\}$, which satisfies the sufficient pair optimality conditions for $\theta=0$ (see Figure~\ref{figexiter}a). Observe that in this case it is enough to check that $X$ is optimal for the costs $C^*_e=C_e$ and $Y$ is optimal for the costs $c^*_e=c_e$, $e\in E$. For $\theta=0$, the admissible graph does not contain any augmenting path. We thus have to modify the costs $C^*_e$ and $c^*_e$, according to~(\ref{newc}). For $\delta^*=1$, a new inequality becomes tight and one arc is added to the admissible graph (see Figure~\ref{figexiter}b). The admissible graph still does not have an augmenting path, so we have to again modify the costs. For $\delta^*=1$ some new inequalities become tight and three arcs are added to the admissible graph (see Figure~\ref{figexiter}c). Now the admissible graph has two augmenting paths (cases~\ref{auc1} and~\ref{auc3}a, see 
the proof of Lemma~\ref{tXVA}). Choosing one of them, and performing the modification described in 
the proof of Lemma~\ref{tXVA} we get a new pair $(X', Y')$ with $|E_{Z'}|=|E_Z|+1$.
  
  Let us now estimate the running time of the procedure. The admissible graph has at most $m$ nodes and at most $mn$ arcs. It can be built in $O(nm)$ time. The augmenting path in the admissible graph can be found in $O(nm)$ time by applying the breath first search. Also the number of inequalities which must be analyzed to find $\delta^*$ is $O(nm)$. Since we have to update the cost of each arc of the admissible graph at most $m$ times, until an augmenting 
  path appears, the required time of the procedure is $O(m^2n)$. We thus get the following result.
\begin{thm}
The \textsc{Rec ST}  problem is solvable in $O(Lm^2n)$ time, where $L=n-1-k$.
\label{trecpol}
\end{thm}

\section{The recoverable robust spanning tree problem}
\label{sec2}

In this section we are concerned with the  \textsc{Rob Rec ST}  problem
under the interval uncertainty representation, i.e.  for the scenario sets
$\mathcal{U}^{I}$,
$\mathcal{U}^{I}_1(\Gamma)$, and $\mathcal{U}^{I}_2(\Gamma)$.
Using the polynomial algorithm for \textsc{Rec ST}, constructed in Section~\ref{sec1},
we will provide a polynomial  algorithm for \textsc{Rob Rec ST} under~$\mathcal{U}^{I}$ and
 some approximation algorithms for a wide class of
 \textsc{Rob Rec ST} under~$\mathcal{U}^{I}_1(\Gamma)$ and $\mathcal{U}^{I}_2(\Gamma)$. 
  The idea will be to solve \textsc{Rec ST} for a suitably chosen second stage costs. Let 
$$F(X)=\sum_{e\in X} C_e + \max_{S\in \mathcal{U}}\min_{Y\in \Phi^k_{X}} f(Y,S),$$
where
$f(Y,S)=\sum_{e\in Y} c^S_e$. It is worth pointing out that
under scenario sets  $\mathcal{U}^{I}$ and $\mathcal{U}^{I}_2(\Gamma)$,
the value of $F(X)$, for a given spanning tree $X$, can be computed in polynomial time~\cite{SAO09,NO13}. 
On the other hand, computing $F(X)$  under  $\mathcal{U}^{I}_1(\Gamma)$ 
turns out to be strongly NP-hard~\cite{NO13, FR99}.
Given scenario $S=(c_e^S)_{e\in E}$, consider the following \textsc{Rec ST} problem:
\begin{equation}
\label{incST1}
\min_{X\in \Phi} \left(\sum_{e\in X} C_e + \min_{Y\in \Phi^k_X} f(Y,S)\right).
\end{equation}
Problem~(\ref{incST1}) is equivalent to the formulation~(\ref{recst}) for $S=(c_e)_{e\in E}$ and it
 is polynomially solvable, according to the result obtained in  Section~\ref{sec1}.
  As in the previous section, we  denote by  pair $(X,Y)$
  a solution to~(\ref{incST1}), where $X\in \Phi$ and $Y\in \Phi^k_X$. Given $S$, we
  call  $(X,Y)$ an \emph{optimal pair under}~$S$  if $(X,Y)$ is an optimal solution to~(\ref{incST1}).
  
The
   \textsc{Rob Rec ST}  problem   with scenario set~$\mathcal{U}^{I}$ can be rewritten as follows:
\begin{align} 
\min_{X\in\Phi}
\left (\sum_{e\in X} C_e  +    \max_{S\in \mathcal{U}^{I}}\min_{Y\in \Phi^{k}_{X}}\sum_{e\in Y}c^{S}_e \right)=
\min_{X\in\Phi}
\left(\sum_{e\in X} C_e  +    \min_{Y\in \Phi^{k}_{X}} \sum_{e\in E}(c_e+d_e)  \right).
\label{irs}
\end{align} 
Thus  (\ref{irs}) is  (\ref{incST1}) for $S=(c_e+d_e)_{e\in E}\in \mathcal{U}^{I}$. Hence and
from  Theorem~\ref{trecpol} we immediately get the following theorem:
\begin{thm}
	For scenario set $\mathcal{U}^I$, the  \textsc{Rob Rec ST}  problem is solvable in 
	$O((n-1-k)m^2n)$ time.
\end{thm}

We now address  \textsc{Rob Rec ST} under~$\mathcal{U}^{I}_1(\Gamma)$ and $\mathcal{U}^{I}_2(\Gamma)$.
Suppose that
 $c_e\geq \alpha(c_e+d_e)$ for each $e\in E$, where $\alpha\in (0,1]$ is a given constant.  
This inequality means that for each edge  $e\in E$ the nominal cost~$c_e$ is positive and $c_e+d_e$  is at most $1/\alpha$ greater than $c_e$. It is reasonable to assume that this condition will be true in many practical applications for not very large value of $1/\alpha$.
\begin{lem}
\label{lemappr1}
         Suppose that
        $c_e\geq \alpha(c_e+d_e)$ for each $e\in E$,  where $\alpha\in (0,1]$, and let
	$(\hat{X}, \hat{Y})$ be an optimal pair under $\underline{S}=(c_e)_{e\in E}$. 
	Then for the scenario sets $\mathcal{U}^{I}_1(\Gamma)$ and $\mathcal{U}^{I}_2(\Gamma)$
	the inequality
	$F(\hat{X})\leq \frac{1}{\alpha} F(X)$ holds for any $X\in \Phi$.
\end{lem}
\begin{proof}
We give the proof only for 
 the scenario set $\mathcal{U}^{I}_1(\Gamma)$.
 The proof for $\mathcal{U}^{I}_2(\Gamma)$ is the same. Let $X\in \Phi$.
  The following inequality is satisfied:
$$F(X)=\sum_{e\in X} C_e+\max_{S\in \mathcal{U}_1^I(\Gamma)}\min_{Y\in \Phi^k_X}f(Y,S)=\sum_{e\in X} C_e+f(Y^*,S^*)\geq \sum_{e\in X} C_e+ f(Y^*,\underline{S}).$$
Clearly, $(X,Y^*)$ is a feasible pair to~(\ref{incST1}) under $\underline{S}$.
 From the definition of $(\hat{X},\hat{Y})$ we get
 \begin{equation}
 F(X)\geq \sum_{e\in \hat{X}} C_e+ f(\hat{Y},\underline{S})=\sum_{e\in \hat{X}} C_e + \sum_{e\in \hat{Y}} c_e \geq \sum_{e\in \hat{X}} C_e + \sum_{e\in \hat{Y}} \alpha(c_e+d_e)= \sum_{e\in \hat{X}} C_e+ \alpha f(\hat{Y},\overline{S}),
 \label{pappro}
 \end{equation}
 where $\overline{S}=(c_e+d_e)_{e\in E}$. Hence
 \begin{align*}
 F(X)&\geq \sum_{e\in \hat{X}} C_e +\alpha \max_{S \in \mathcal{U}_1^I(\Gamma)}f(\hat{Y},S) 
 \geq \sum_{e\in \hat{X}} C_e +\alpha \max_{S \in \mathcal{U}_1^I(\Gamma)}\min_{Y\in \Phi^k_{\hat{X}}}f(Y, S) \\
&\geq \alpha \left(\sum_{e\in \hat{X}} C_e + \max_{S \in \mathcal{U}_1^I(\Gamma)}\min_{Y\in \Phi^k_{\hat{X}}}f(Y, S)\right)
 =\alpha F(\hat{X})
 \end{align*}
 and the lemma follows.
\end{proof}
The condition  $c_e\geq \alpha(c_e+d_e)$, $e\in E$,
in Lemma~\ref{lemappr1},
 can be weakened and, in consequence,
the set of instances to which the approximation ratio of the algorithm applies can be extended. Indeed,  from inequality~(\ref{pappro})
it follows that the bounds of the uncertainty intervals are only required to meet 
the condition $\sum_{e\in \hat{Y}} c_e \geq  \alpha\sum_{e\in \hat{Y}}(c_e+d_e)$. This condition can be verified efficiently, since $\hat{Y}$ can be computed in polynomial time.

We now focus on  
 \textsc{Rob Rec ST} for~$\mathcal{U}^{I}_2(\Gamma)$.
Define $D=\sum_{e\in E} d_e$ and suppose that $D>0$ (if $D=0$, then the problem is equivalent to \textsc{Rec ST} for the second stage costs $c_e$, $e\in E$). 
Consider scenario $S'$ under which $c_e^{S'}=\min\{c_e+d_e, c_e+\Gamma\frac{d_e}{D}\}$ for each $e\in E$. 
Obviously,
$S'\in \mathcal{U}^{I}_2(\Gamma)$, since 
$\sum_{e\in E} \delta_e \leq \sum_{e\in E} \Gamma\frac{d_e}{D}\leq \Gamma$.  
The following theorem provides another approximation result for
 \textsc{Rob Rec ST} with scenario set~$\mathcal{U}^{I}_2(\Gamma)$:
\begin{lem}
\label{lemappr2}
	Let $(\hat{X}, \hat{Y})$ be an optimal pair under~$S'$. Then the following implications are true for scenario set $\mathcal{U}^{I}_2(\Gamma)$:
	\begin{itemize}
 \item[(i)]
   If $\Gamma \geq \beta D$, $\beta\in (0,1]$, then  $F(\hat{X})\leq \frac{1}{\beta}F(X)$ for any $X\in \Phi$.
 \item[(ii)] If $\Gamma \leq \gamma F(\hat{X})$, $\gamma\in [0,1)$ then $F(\hat{X})\leq \frac{1}{1-\gamma}F(X)$ for any $X\in \Phi$.
	\end{itemize}
\end{lem}
\begin{proof}
Let $X\in \Phi$. Since $S'\in \mathcal{U}^{I}_2(\Gamma)$, we get
 \begin{equation}
 \label{ee0}
 F(X)= \sum_{e\in X} C_e+\max_{S\in \mathcal{U}_2^I(\Gamma)}\min_{Y\in \Phi^k_X}f(Y,S)\geq \sum_{e\in X} C_e+\min_{Y\in \Phi^k_{X}}f(Y,S').
 \end{equation}
 We first prove implication~$(i)$. By~(\ref{ee0}) and  the definition of $(\hat{X},\hat{Y})$, we obtain
 \begin{align*}
 F(X)&\geq \sum_{e\in \hat{X}} C_e+f(\hat{Y},S') =\sum_{e\in \hat{X}} C_e+ \sum_{e\in \hat{Y}} \min\{c_e+d_e, c_e+\Gamma \frac{d_e}{D}\}\\
&\geq\sum_{e\in \hat{X}} C_e+ \sum_{e\in \hat{Y}} \min\{c_e+d_e, c_e+\beta d_e\}=
 	\sum_{e\in \hat{X}} C_e+ \sum_{e\in \hat{Y}}  (c_e+\beta d_e)\geq \sum_{e\in \hat{X}} C_e+ \beta f(\hat{Y},\overline{S}),
 \end{align*}
 where $\overline{S}=(c_e+d_e)_{e\in E}$. The rest of the proof is the same as in the proof of Lemma~\ref{lemappr1}.
  We now prove implication~$(ii)$. By~(\ref{ee0}) and the definition of $(\hat{X},\hat{Y})$, we have
\begin{align*}
  F(X)&\geq \sum_{e\in \hat{X}} C_e+f(\hat{Y},S') \geq \sum_{e\in \hat{X}} C_e+f(\hat{Y},\underline{S})\geq
 	\sum_{e\in \hat{X}} C_e+ \max_{S\in \mathcal{U}_2^I(\Gamma)}f(\hat{Y},S)-\Gamma\\
&\geq\sum_{e\in \hat{X}} C_e+ \max_{S\in \mathcal{U}_2^I(\Gamma)} \min_{Y\in \Phi^k_{\hat{X}}}f(Y,S)-\Gamma
=F(\hat{X})-\Gamma.
\end{align*}
If $\Gamma \leq \gamma F(\hat{X})$. Then
$F(X)\geq F(\hat{X})-\gamma F(\hat{X})=(1-\gamma)F(\hat{X})$
and $F(\hat{X})\leq \frac{1}{1-\gamma} F(X).$
\end{proof}
Note that the value of $F(\hat{X})$ under $\mathcal{U}^{I}_2(\Gamma)$ can be computed in 
polynomial time~\cite{NO13}.
In consequence, the constants $\beta$ and $\gamma$ can be efficiently determined for every particular instance of the problem.
Clearly, we can assume that $d_e\leq \Gamma$ for each $e\in E$, which implies  $D\leq m\Gamma$, where $m=|E|$.  Hence, we can assume that $\Gamma\geq \frac{1}{m} D$ for every instance of the problem.  We thus get from Lemma~\ref{lemappr2} (implication (i)) that $F(\hat{X})\leq m F(X)$ for any $X\in \Phi$ and the problem is approximable within $m$.
If $\alpha$, $\beta$ and $\gamma$ are the constants from Lemmas~\ref{lemappr1} and~\ref{lemappr2}, then 
 the following theorem summarizes the approximation results:
\begin{thm}
	\textsc{Rob Rec ST}  is approximable within $\frac{1}{\alpha}$ 
	under scenario set $\mathcal{U}^{I}_1(\Gamma)$  
	and  it is approximable within $\min\{\frac{1}{\beta},\frac{1}{\alpha},\frac{1}{1-\gamma}\}$
	under scenario set $\mathcal{U}^{I}_2(\Gamma)$.
\end{thm}

 Observe that Lemma~\ref{lemappr1} and Lemma~\ref{lemappr2}  hold of any sets $\Phi$ and $\Phi^k_X$ (the particular structure of these sets is not exploited). Hence the approximation algorithms can be applied to any problem for which the recoverable version~(\ref{incST1}) is polynomially solvable.

\section{Conclusions}

In this paper we have studied the recoverable robust spanning tree problem (\textsc{Rob Rec ST}) 
under various  interval uncertainty representations. The main result is the polynomial time combinatorial algorithm for the recoverable spanning tree.
We have applied   this  algorithm 
for solving   \textsc{Rob Rec ST} under the traditional uncertainty representation (see, e.g.,~\cite{KY97})
 in polynomial time. 
Moreover, we have used the algorithm  for providing several approximation results for \textsc{Rec ST}  with the scenario set introduced by Bertsimas and Sim~\cite{BS03}  and
the scenario set with a budged constraint (see, e.g,.~\cite{NO13}).
There is a number of open questions concerning the considered problem. Perhaps, the most interesting one is to resolve the complexity of the robust problem under the interval uncertainty representation with budget constraint. It is possible that this problem may be solved in polynomial time by some extension of the algorithm constructed in this paper. One can also try to extend the algorithm for the more general recoverable matroid base problem, which has also been shown to be polynomially solvable in~\cite{HKZ16}.

\subsubsection*{Acknowledgements}
The second and the third authors were
supported by
 the National Center for Science (Narodowe Centrum Nauki), grant  2013/09/B/ST6/01525.


\end{document}